\documentclass[sigconf,nonacm]{acmart}

\AtBeginDocument{%
  }

\usepackage{tikz}
\usepackage{graphicx}
\usepackage{subfigure}
\usepackage{multirow}
\usepackage{array}
\usepackage{caption}
\usepackage{soul}
\usepackage{pgfplots}
\usepackage{pgfplotstable}
\usepgfplotslibrary{statistics}
\usepgfplotslibrary{groupplots}
\usepgfplotslibrary{colorbrewer}
\usepackage{url}
\usepackage{hyperref}
\usepackage[shortlabels]{enumitem}
\usepackage{xspace}
\usepackage{gensymb}

\usepackage[linesnumbered,ruled,vlined]{algorithm2e}

\usepackage{xcolor}
\usepackage{mdframed}
\definecolor{mygray}{RGB}{211,211,211} 

\newcommand{\tool}{{\sc CLUE-Mark}\xspace}
\newcommand{\treering}{{\sc Tree Ring}\xspace}
\newcommand{\gs}{{\sc Gaussian Shading}\xspace}
\renewcommand{\vec}{\mathbf}
\newcommand{\reals}{\mathbb{R}}
\newcommand{\gaussian}{\mathcal{N}(\vec{0}, \vec{I})}
\newcommand{\keyspace}{\mathcal{K}}
\newcommand{\tokenspace}{\mathcal{T}^*}

\newcommand{\sourceref}{%
A proof-of-concept implementation of \tool{} and source
code to reproduce our evaluations is available at:
\url{https://github.com/kisp-nus/cluemark}%
}

\newtheorem{theorem}{Theorem}
\newtheorem{lemma}{Lemma}[theorem]
\newtheorem{claim}{Claim}[theorem]
\newtheorem{assumption}{Assumption}[theorem]
\theoremstyle{definition}
\newtheorem{definition}{Definition}[section]

\begin{document}

\title{\tool: Watermarking Diffusion Models using CLWE}

\author{Kareem Shehata}
\affiliation{%
  \institution{School of Computing, National University of Singapore}
  \city{Singapore}
  \country{Singapore}
}
\email{kareem@u.nus.edu}

\author{Aashish Kolluri}
\authornote{Work done while at National University of Singapore.}
\affiliation{%
  \institution{Azure Research}
  \city{Cambridge}
  \country{UK}
}
\email{aashishk@u.nus.edu}

\author{Prateek Saxena}
\affiliation{%
  \institution{School of Computing, National University of Singapore}
  \city{Singapore}
  \country{Singapore}
}
\email{dcsprs@nus.edu.sg}

\renewcommand{\shortauthors}{Shehata et al.}

\begin{abstract}

  The proliferation of AI-generated images has led to a new problem
known as ``AI Slop'': low-effort, mass-produced AI-generated images that are distributed as-is but are
difficult to distinguish from authentic images. Watermarking such content
from popular diffusion models provides a promising solution: if
social media platforms and end users can easily identify images
generated by the most common AI providers, then a very large
proportion of the AI Slop problem can be detected by end users. However, prior
watermarking techniques are heuristic and lack formal guarantees of
undetectability. Techniques that were claimed to not affect image
quality were later shown to in fact degrade quality, hampering
adoption as providers and users are averse to any reduction in
output quality.

In this work, we introduce \tool{}, a provably undetectable
watermaking scheme for diffusion models. \tool{} requires no changes
to the model being used, is computationally efficient, and guaranteed
to have no impact on model output quality. Our approach leverages the
Continuous Learning With Errors (CLWE) problem --- a
cryptographically hard lattice problem --- to embed hidden messages
in the latent noise vectors used by diffusion models. We reduce the
security of CLWE against an adversary given limited samples to the
security of \tool{}. Assuming that CLWE is secure in the limited
sample regime, image watermarks generated by \tool{} are undetectable
by {\em any} efficient adversary that does not have the secret key,
making them imperceptible to human observers and undetectable by
ad-hoc heuristics as well. \tool{} focuses primarily on
undetectability while maintaining sufficient robustness for common
non-adversarial transformations such as JPEG compression. Empirical
evaluations on state-of-the-art diffusion models confirm that \tool{}
achieves high message recovery, preserves image quality, and is
robust to minor perturbations such as JPEG compression and brightness
adjustments. \sourceref{}

\end{abstract}




\maketitle

\section{Introduction}
\label{sec:intro}

The internet has recently witnessed an unprecedented surge of
AI-generated content, an effect that was first popularly referred to
as ``AI Slop'' in mid-2024~\cite{hern2024slop, nytimes_slop}. This
phenomenon has now entered popular discourse with coverage from diverse media
outlets including The Guardian~\cite{hern2024slop, malik2025slop,
mahdawi2025aislop}, The New York Times~\cite{nytimes_slop}, New York
Magazine Intelligencer~\cite{read2024drowningInSlop}, and John
Oliver's Last Week Tonight~\cite{oliver2025aislop}. 
As pointed out by many commentators, the problem with such content
 is not the generation of content itself, but
that for many users it is difficult to distinguish authentic and
synthetic content, leading to
problems of misinformation and user trust.

Such machine-generated content is created in large volumes and shared through social media
with little to no modification. Many users feel either frustrated by
the inability to distinguish it from authetic content or overwhelmed
by the sheer volume of it, or both, with many blaming social media
platforms for doing too little to identify AI generated
content~\cite{hern2024slop, malik2025slop}. Thus, the problem of
identifying AI-generated images is both timely and important.

One proposed method to identifying AI-generated images is
watermarking, in which the image is ``marked'' in such a way that it
can be distinguished from an unmarked image. In this manner, should
publicly available AI image generation tools embed watermarks in
their output, as has been proposed by several legal
projects~\cite{ai_watermarking_act, uswatermark, pledge,
euwatermark}, it would be possible to identify AI-generated content
automatically either by social media platforms or by end-user
software. As the name and physical analogy suggests, watermarking is
typically designed to be robust to benign transformations, but is
also easily detectable by the user, producing at least some visible
artifacts or other degradation in image quality. This presents a
major barrier to adoption among both providers and users of AI image
generation tools, as they are (understandably) reluctant to degrade
the quality of the images generated by their models after spending
considerable effort to train them.

Prior approaches to watermarking AI-generated images were ad-hoc. For
example, \treering~\cite{treering} was conjectured to be undetectable
and robust, but no proof as such was given. While an important step
forward, \treering watermarks can be both detected and removed by
standard steganalysis~\cite{steg_attack}, i.e.\ averaging the
difference between watermarked and non-watermarked images. Similarly,
\gs~\cite{gaussian_shading} was proved to be undetectable, but only
under the assumption that the nonce used to generate the watermark is
not reused; otherwise the watermark is again detectable by
steganalysis. The one-time (per image) nonce assumption is hard to satisfy practically. This is more than a security concern; it shows that the
watermark affects the output of the model such that the output images
are both quantitatively and qualitatively different from the original
(a finding that we verify in Section~\ref{sec:eval-steganographic}).

Thus for both security and quality reasons, it is important that the
watermark be \emph{provably undetectable}, by which we mean that a
watermarked image must be indistinguishable from an unwatermarked
image\footnote{ Embedding a signal in an image such that it cannot be
detected but is not robust to perturbations is technically considered
\emph{steganography} rather than watermarking, but but the latter
term remains popular.}. A rigorous reduction from a cryptographically
hard problem is needed, otherwise it may turn out that a watermark
that was conjectured to be undetectable is in fact detectable by some
other means, or may affect downstream applications. By proving
undetectability, we ensure that the watermark does not compromise the
quality of the generated images and remains stealthy to detectors. In
an ideal world, such an undetectable watermark (properly:
steganography) would also be robust to benign transformations.
Unfortunately, it has been proved that a perfect watermark cannot
exist (see Appendix~\ref{sec:robust} and \cite{wm_impossibility}),
and thus we must choose between robustness and undetectability.

In this work we present \tool{}, an undetectable watermark (i.e.\
steganography) scheme for diffusion models to address the problem of
AI Slop. Given the challenges listed above and the impossibility of
perfect watermarking, we focus on using cryptographic techniques to
inject hidden information into diffusion model generated images such
that any efficient algorithm that can distinguish marked and unmarked
images can be used to break fundamental cryptographic assumptions ---
something assumed to be impossible --- while verifiers that have
access to the secret key can recover the hidden information with high
probability, despite minor purturbations such as JPEG compression.
This allows labelling of AI-generated content at scale, precise the
problem posed by AI Slop. Even if sophisticated users employing more
effort or advanced techniques can remove the hidden watermark, this
prevents low-effort users from flooding social media networks with
unmarked AI-generated content, drastically improving the scale of the
problem. By proving undetectability, we ensure that our method does
not compromise the quality of the generated images.

\paragraph{Undetectability from CLWE} Our approach modifies the noise
input to a diffusion model using a distribution based on the
Continuous Learning With Errors (CLWE) problem~\cite{clwe}. Intuitively, this
distribution is gaussian (i.e.\ normal) in all directions except for
a secret direction in which it has a signature pattern. For our \tool{} construction, we can prove that any process
that can reliably detect this hidden pattern can be used to
distinguish a CLWE distribution from a normal distribution. Assuming that this problem is cryptographically hard, we can guarantee
that our approach is undetectable by \emph{any} efficient algorithm
without the secret key, ensuring perfect image quality and
compatibility with all content creation tools. Because our technique
only modifies the input noise, it does not require any model
modifications or retraining.

\paragraph{Empirical Results} While \tool{} aims for undetectability, we
must also demonstrate its practical utility. Can our steganographic signal be
reliably recovered after image generation and common non-adversarial
transformations like JPEG compression? To answer this, we evaluate \tool{} on
commonly used diffusion models and datasets to verify undetectability,
recoverability, image quality, and robustness to non-adversarial
transformations. We show that \tool{} is undetectable by standard means, produces
images with identical quality to those without steganographic content, and the
hidden message can be reliably recovered after JPEG compression and minor
brightness adjustments. \sourceref{}

\paragraph{Comparison to PRC-based Watermarks} In work contemporary
to ours, Gunn et al \cite{prc_watermarks} propose using Pseudo-Random
Codes (PRC) to bias the latent input to a diffusion model for watermarking. 
While PRC offer a promising new primitive, our work uses CLWE---a complementary primitive to encode secret information in the {\em continuous $\mathbb{R}^d$ domain} directly---which is more natural fit for machine learning systems. PRC-based watermarks use parameters that are not known to be theoretically secure, as stated in their work, and is more suitable for applications where robustness to adversarial perturbation is more important than undetectability.
We discuss in more
detail how our work differs from PRC watermarks in
Section~\ref{sec:related}.

\paragraph{Contributions.} Our contribution is a novel and practical
watermarking scheme, \tool{}, for images generated by modern
diffusion AI models. It aims to be undetectable, in the sense that
distinguishing marked and unmarked images is as hard as
distinguishing the CLWE distribution from a Gaussian. This assumption is stricter than the standard cryptographic hardness of CLWE; it considers the practically-motivated regime where the adversary has limited samples.
Along the way we develop a novel hCLWE sampling algorithm (see
Section~\ref{sec:sampling_hclwe}) that is more practical
than the previous rejection-sampling technique proposed
in~\cite{clwe}, which may be of independent interest. These tools
enables detection of low-effort AI-generated content (``AI Slop'')
without compromising image quality or compatibility with existing
tools and platforms.

\section{Motivation and Problem Setup}%
\label{sec:motivation-problem}

\begin{figure*}[!ht]
    \centering
    \includegraphics[width=0.9\textwidth,trim={0 21cm 16cm 0},clip]{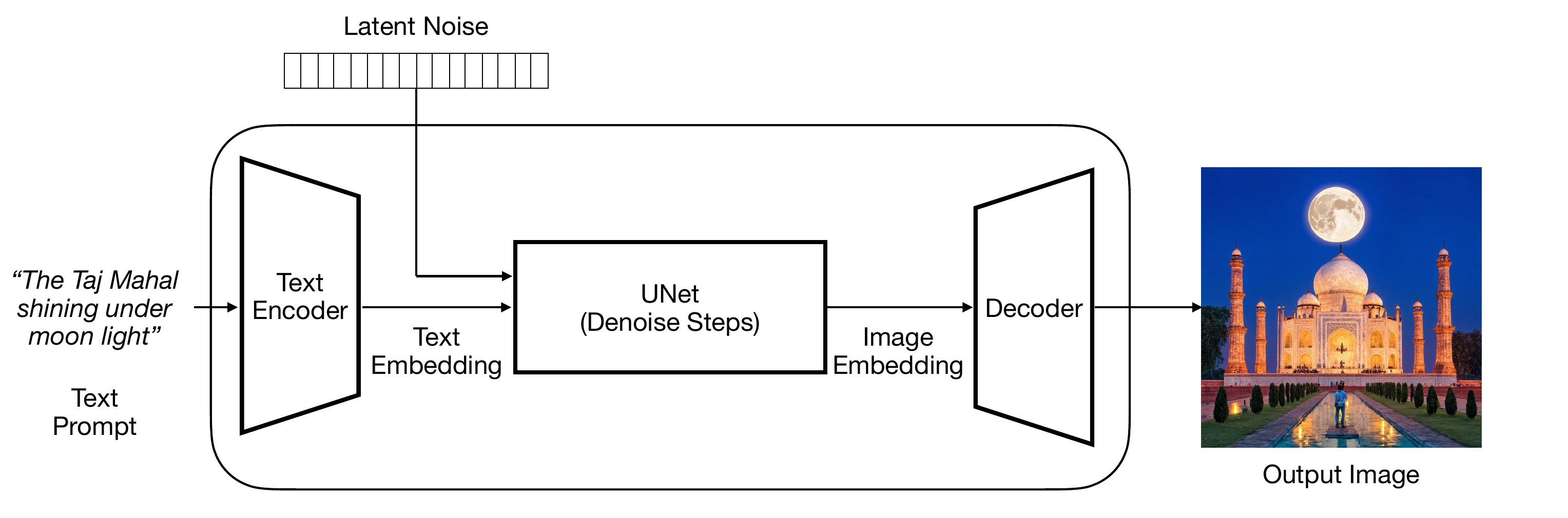}
    \caption{Sampling images from latent diffusion models}%
    \label{fig:ldm}
\end{figure*}

Diffusion models, such as Stable
Diffusion~\cite{podell2023sdxl,rombach2022high}, have quickly emerged
as the go-to AI models for generating high quality and
customizable images. They have been adopted by content creators across diverse fields,
from digital art to professional illustration for e-commerce and advertising 
(see survey~\cite{yang2023diffusion}).
The process to generate an image from a diffusion model is called
``sampling,'' in which the model is given a text embedding and a
noise vector (in the latent space of the model, hence this is often
called the ``latent vector'' or ``latent noise''), and the model
iteratively removes the noise to produce an image described by the
text embedding, as illustrated in Figure~\ref{fig:ldm}. We refer
to~\cite{rombach2022high, podell2023sdxl} for understanding the
details of diffusion models, and describe only the parts that are
necessary for our work.

\begin{definition}[Diffusion Model Sampling]
\label{def:diff_sampling}

    The process to sample a diffusion model is defined as the function:
    
    \[
        \text{Sample}_\theta(\pi \in \tokenspace; \vec{z} \in \reals^n )
        \rightarrow \vec{x} \in \reals^N
    \]
    
    where $\theta$ is the model parameters, $\pi$ is the embedding of
    the text prompt, $\vec{z}$ is the noise vector, and $\vec{x}$ is
    the generated image. The noise vector is drawn from a standard
    normal (Gaussian) distribution, that is $\vec{z} \sim \gaussian$.
    
\end{definition}


\paragraph{Watermarking Diffusion Models.} Watermarking is a proposed
measure to detect AI-generated content, which can be used to fight
against AI Slop. Watermarking an image refers to implanting a signal
within it such that the signal can later be recovered, providing
proof of the provenance of the image. In a secret key scheme, only
those with the key can verify an image's origin, which is technically
closer to steganography than watermarking, the latter usually
prioritizes robustness to perturbations over undetectability.
Consider a company that hosts a closed-source model. If the company
watermarks the images generated by the model using a secret key, they
can later verify whether an image was generated by the model by
checking for the watermark using the key. They can also give the key
to another entity to do verification, for example a social media
platform that wants to verify the provenance of images uploaded by
users. The company may also assign a key to each user in order to
trace which user was responsible for generating a particular image.
Notice that the verifier, whether the model provider, social media platform,
or another party, is given the model keys but not any per-image data,
such as text prompts, nonces, or noise vectors.
Watermarks may be applied to images after generation
(``post-process'') or during model execution (``in-process'').
Diffusion models offer another option: watermarking the noise vector
$\vec{z}$ used to generate the image (``pre-process''), as introduced
by Tree Rings~\cite{treering}. In this work we focus on secret-key
pre-process watermarking. See Figure~\ref{fig:wm_use_cases} for an
illustration of the use cases.

\begin{figure*}[t]
    \centering
    \includegraphics[width=0.8\textwidth]{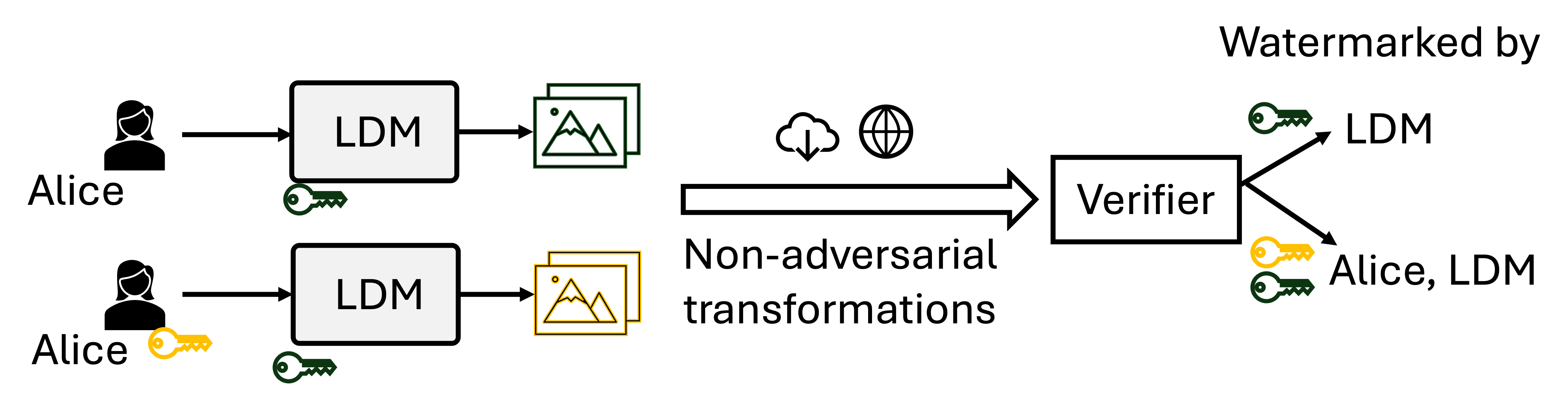}
    \caption{Watermarking use cases and our threat model, with a per-model (and per-creator) keys. Note that
    the verifier is given model keys, but no image-specific data, 
    such as text prompts, noise vectors, or per-image nonces or keys.}
    \label{fig:wm_use_cases}
\end{figure*}

\paragraph{Previous approaches} Despite the potential benefits of
watermarking AI generated images, adoption has been
limited~\cite{wsj_wm}. One reason for this is image quality. Large
models that generate high-quality images are expensive to train and
serve, and both companies and users are sensitive to any degradation
in image quality. Further, changes to image quality may also affect
the performance of downstream tasks. For example, if the generated
images are used to train another model or as inputs to a larger
process, any changes in the images may affect the performance of the
system as a whole.

Existing watermarking techniques for diffusion models have a
noticeable impact on the images generated. Post-processing
techniques, e.g.\ embedding a signal into the output pixel data and
even AI-based techniques, have visible quality
losses~\cite{treering}. As we evaluate in this work, state-of-the-art
techniques such as \treering and \gs, despite aiming for not being
perceptible, produce images with a significant quality differences
from their unwatermarked equivalents (see Section~\ref{sec:quality}).
This is because of ad-hoc design choices and the lack of provable
guarantees. When evaluated over a large set of images, it can be seen
that the watermarked images are biased (less diverse statistically)
than the unwatermarked images. Further, recent attacks have shown
that the watermarks can be detected, removed, and forged using
standard steganographic techniques~\cite{steg_attack}.

\paragraph{Problem} The question thus becomes: can we watermark
diffusion model generated images such that a verifier as in
Figure~\ref{fig:wm_use_cases} can distinguish between watermarked and
unwatermarked images, while proving that the watermark does not
affect the quality of the images? We formally define this problem in
the next section.

\subsection{Undetectable Watermarking}
\label{sec:problem}

In this work, we aim to solve this problem by designing {\em provably
undetectable} watermarking. By this we mean that distinguishing
between watermarked and unwatermarked images must be computationally
at least as hard as solving a cryptographic problem. This ensures
that there is a negligible difference between watermarked and
unwatermarked images for \emph{any} downstream application. Even a
determined adversary trying just to detect the presence of a
watermark faces a challenge at least as hard as breaking
state-of-the-art cryptography. This approach embodies a secure by
construction principle, in contrast to existing techniques which
empirically evaluate against specific attacks but cannot rule out the
existence of other effective attacks. In order to construct a
provable scheme, we first formally define the watermarking process.

\begin{definition}[Watermarking]
\label{def:watermarking}

    A watermarking function is defined as the tuple of functions
    $(\text{Setup}, \text{Mark}$, $\text{Extract})$:

    \[
        \text{Setup}_\theta(1^\lambda) \rightarrow k \in \keyspace
    \]
    \[
        \text{Mark}_\theta(\pi \in \tokenspace; k \in \keyspace)
        \rightarrow \vec{x} \in \reals^N
    \]
    \[
        \text{Extract}_\theta(\vec{x} \in \reals^N; k \in \keyspace)
        \rightarrow \{0, 1\}
    \]

    where $\lambda$ is a security parameter, $\theta$ is the model
    parameters, $\pi$ is the embedding of the text prompt, $k$ is the
    secret key, and $\vec{x}$ is the generated image.

\end{definition}

For a watermarking scheme to be useful, it needs to satisfy both
completeness and soundness criteria, for which we use
Definitions~\ref{def:wm_complete} and~\ref{def:wm_sound} below.
Notice that the definitions are worst-case over all prompts. Our
approach currently focuses on embedding a single bit of information,
though this can be extended to multiple bits through phase-shifting
techniques or using multiple keys.

\begin{definition}[Watermark Completeness]%
\label{def:wm_complete}

    A watermarking scheme $(\text{Setup}, \text{Mark}$, $\text{Extract})$ is
    $\delta$-complete if for all $\pi \in \tokenspace$:
    \[
        \Pr\left[ \text{Extract}_\theta(
            \text{Mark}_\theta(\pi; k); k) = 1
            | k \leftarrow \text{Setup}_\theta(1^\lambda)
        \right]
        \ge 1 - \delta
    \]

    Where the probability is taken over all random coins used in
    Setup, Mark, and Extract.

\end{definition}

\begin{definition}[Watermark Soundness]%
\label{def:wm_sound}

    A watermarking scheme $(\text{Setup}, \text{Mark}$, $\text{Extract})$ is
    $\delta$-sound if for all $\pi \in \tokenspace$:
    \begin{align*}
        \Pr[ \text{Extract}_\theta(
            &\text{Sample}_\theta(\vec{z}; \pi); k) = 1
            | \\
            &k \leftarrow \text{Setup}_\theta(1^\lambda),
            \vec{z} \leftarrow \gaussian]
        \le \delta
    \end{align*}

    Where the probablity is taken over $\vec{z}$ and all random
    coins used in Setup.
\end{definition}

Before we can define undetectability, we must define formally what we
mean by ``distinguishing'' and how we measure success. For this we
use Definition~\ref{def:adv}.

\begin{definition}[Advantage of a Distinguisher]%
\label{def:adv}

Let $D^{\mathcal{O}} \rightarrow \{0,1\}$ be a distinguisher with
access to oracle $\mathcal{O}$. Define the advantage of $D$ in
distinguishing $\mathcal{O}_1$ and $\mathcal{O}_2$ as:
\begin{align*}
    Adv_{D, \mathcal{O}_1, \mathcal{O}_2} =
    \left| \Pr\left[ D^{\mathcal{O}_1}(1^\lambda) = 1 \right]
      - \Pr\left[ D^{\mathcal{O}_2}(1^\lambda) = 1 \right] \right|
\end{align*}

\end{definition}

We want to ensure that our watermarking scheme is
\emph{undetectable}, that is, that the output of the watermarking
process is indistinguishable from non-watermarked images generated by
the model. For both users and providers, this guarantees that
watermarked outputs have identitical quality to unmarked outputs, and
remains compatible with all existing tools and platforms. We formally
define this in Definition~\ref{def:wm_undetectable}.

\begin{definition}[Provable Undetectability]%
\label{def:wm_undetectable}

    Let $\mathcal{M}_{\theta}$ be an oracle that on input $\pi \in
    \tokenspace$ samples $\vec{z} \leftarrow \gaussian$ and returns
    $\text{Sample}_\theta(\pi; \vec{z})$. Let $\mathcal{S}_{ \theta,
    k}$ be an oracle that on input $\pi \in \tokenspace$ returns
    $\text{Mark}_\theta(\pi; k)$ for a fixed key $k \leftarrow
    \text{Setup}_\theta(1^\lambda)$. A watermarking scheme
    $(\text{Setup}, \text{Mark}$, $\text{Extract})$ is undetectable
    if for all {\em probabilistic polynomial-time\/} PPT distinguishers
    $D$, $Adv_{D, \mathcal{M}_{\theta}, \mathcal{S}_{ \theta, k}} \le
    negl(\lambda)$.
    
\end{definition}

While robustness against adversarial manipulations is a secondary
concern in our AI Slop prevention use case, we still require the
embedded information to survive common non-adversarial
transformations that occur in normal usage, such as JPEG compression
when images are shared on social media platforms (see Figure~\ref{fig:wm_use_cases}). Our approach
prioritizes absolute undetectability first, while maintaining
sufficient robustness for practical usage.

\section{The \tool{} Approach}%
\label{sec:approach}

Our approach to steganography is based on the idea of modifying
the latent noise vector of a diffusion model (i.e.\ pre-process
steganography). Recall that a diffusion model takes as input a latent
vector that is expected to be from the standard gaussian
distribution. In our approach, we instead use a latent vector that is
drawn from the hCLWE distribution (which will be defined in
Section~\ref{sec:clwe}). This vector is indistinguishable from the
standard gaussian distribution to all PPT-computable functions
without the secret key, but can be efficiently distinguished when the
secret key is known. We call using hCLWE samples to generate an image
the {\em signal injection\/} process. To test whether an image contains
a signal, a system with access to the secret key uses the {\em
signal recovery\/} process. During recovery, we invert the diffusion
model to get an estimate of the initial latent vector from the
image. We then use the secret key to determine if the
latent vector is from the hCLWE distribution, and if so, then it is
very likely that the image contains a signal using that key.

In this section we describe the high-level approach of steganography, including
the necessary background on CLWE, and the signal injection and recovery
processes. We will prove undetectability of the construction in
Section~\ref{sec:security_proof}.

\subsection{Background on CLWE}%
\label{sec:clwe}

Consider the scatter plot shown in Figure~\ref{fig:clwe_scatter}. If
you were to look at the distribution of the points along either axis,
it would look like a standard gaussian distribution. But if you were
to plot the points along the line $y = -x$, you would notice that the
points are periodic. That is, that they're bunched together around
$-1, -1/2, 0, 1/2, 1$, and so on. In two dimensions you can see this
immediately by looking at the scatter plot, but what if the points
are in three dimensions? The points would form a sphere from nearly
any angle you look at, except if you found the ``secret direction''
you would notice that the points form slices and are again
periodic\footnote{The original CLWE presentation at
\url{https://youtu.be/9CB4KcoRB9U?t=60} has an animation that
illustrates this effect clearly.}.

This is the essence of the homogenous Continuous Learning With Errors
(hCLWE) problem\footnote{Homogenous CLWE is a special case of the
more general CLWE problem.}. Intuitively, it becomes harder to find
the secret direction as you increase the number of dimensions, as the
number of directions you have to check grows exponentially. Deciding
whether a set of points (from now on let us call them samples) are
normally distributed in all directions, or if there is a secret
direction in which they are periodic has been proven to be as hard as
solving worst-case lattice problems~\cite{clwe}.

More formally, consider a unit vector $\vec{w} \in
\reals^n$ chosen uniformly at random, and a set of samples
$\vec{y}_1, \ldots, \vec{y}_m \in \reals^n$ that are drawn from a
standard gaussian distribution in all directions except for the
direction $\vec{w}$, in which case the distribution is shown in
Figure~\ref{fig:hclwe_pdf}. This produces the well-known
``gaussian-pancakes'' distribution, which produces periodic slices
with separation $\approx 1/\gamma$ and width $\approx \beta / \gamma$
in the secret direction $\vec{w}$, and standard gaussian in all
others. The hCLWE problem is to determine the secret direction
$\vec{w}$ given a polynomial number of samples ${\left\{\vec{y}_i\right\}}_{i \in
[m]}$. The decision variant of hCLWE is to distinguish between
samples from the hCLWE distribution and from the standard
multivariate normal. This leads to Theorem~\ref{thm:hclwe}, an
informal statement of the CLWE hardness theorem; see~\cite{clwe} for
the more formal statement, proof, and further details.

\begin{theorem}[hCLWE Hardness Informal]%
\label{thm:hclwe}

    For $\gamma = \Omega (\sqrt{n}), \beta \in (0, 1)$ such that
    $\beta / \gamma$ is polynomially bounded, if there exists an
    algorithm that can efficiently distinguish the
    $\text{hCLWE}_{\gamma,\beta}$ distribution from the standard
    multivariate gaussian distribution, then there exists an
    efficient algorithm that approximates worst-case lattice problems
    to within a factor polynomial in $n$.

\end{theorem}

Such lattice problems are considered cryptographically hard for
appropriately chosen parameters $\gamma$ and $\beta$. Assuming that
no PPT algorithm can solve them, then we can state that no PPT
algorithm can solve the decision hCLWE problem with non-negligible
probability. Assumption~\ref{def:hclwe_assumption} given later in Section~\ref{sec:formal_defs} is the more precise version we rely on. It constraints the number of samples available.

\begin{figure}
\centering
\includegraphics[width=0.9\columnwidth]{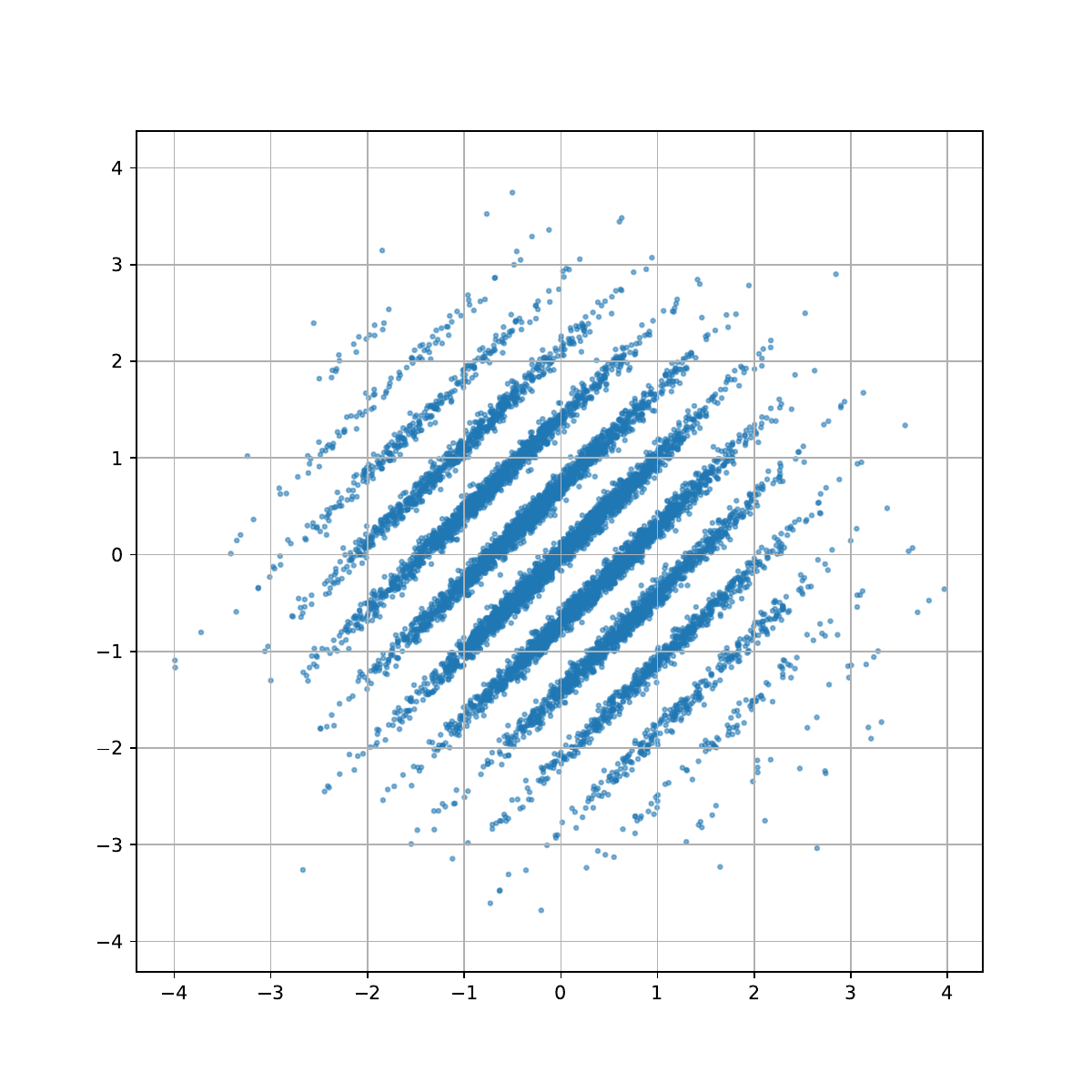}
\caption{Example scatter plot of $10,000$ samples from the hCLWE distribution
with secret direction $(1/\sqrt{2}, -1/\sqrt{2})$, $\gamma = 2, \beta=0.1$.}%
\label{fig:clwe_scatter}
\end{figure}

\begin{figure}
\centering
\includegraphics[width=\columnwidth]{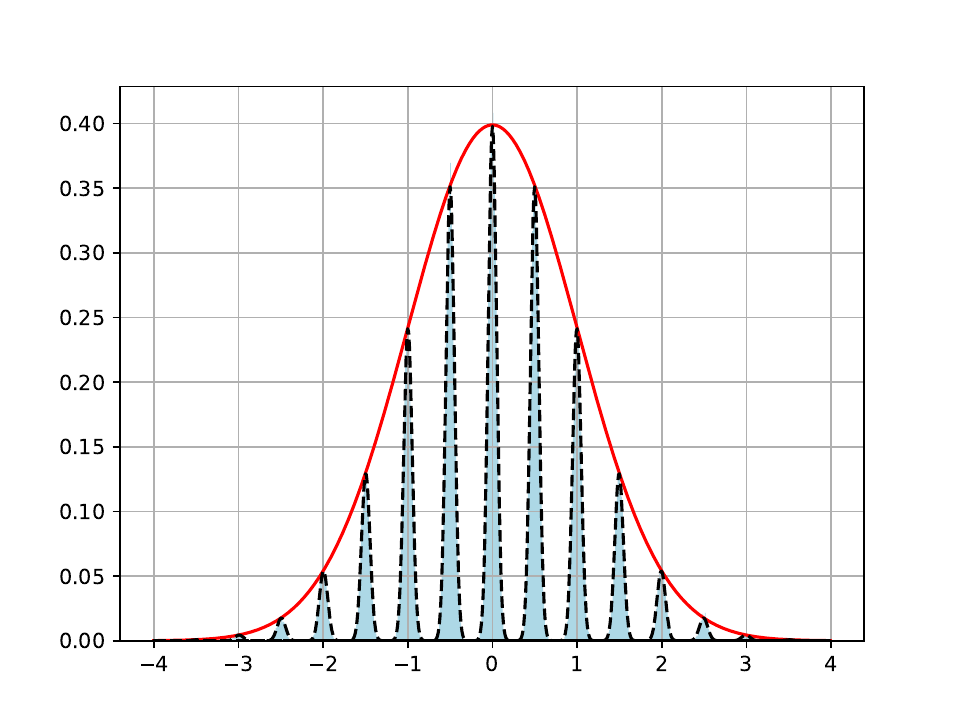}
\caption{The hCLWE PDF with $\gamma=2, \beta=0.1$ in the secret direction
rescaled for comparison with a standard gaussian.}%
\label{fig:hclwe_pdf}
\end{figure}

\subsection{Creating Watermarks from CLWE}%
\label{sec:clwe_wm}

\begin{figure*}[ht]
\centering
\includegraphics[width=0.9\textwidth,trim={0 23cm 8cm 0},clip]{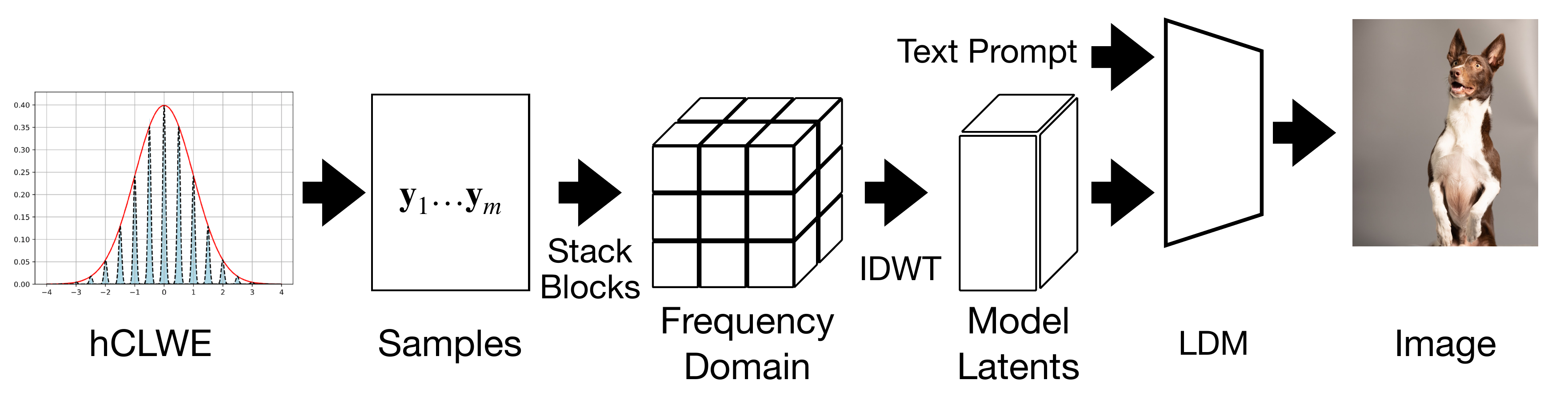}
\caption{\tool{} signal injection process showing how the latent vector is divided into blocks and transformed using the inverse discrete wavelet transform.}%
\label{fig:steganography}
\end{figure*}

There are several technical challenges we overcome in making CLWE work reliably for diffusion models, which are explained next. 

\paragraph{Blocking.} In order to use samples from the hCLWE
distribution as inputs to a diffusion model, we must first overcome a
small disparity: The hCLWE distribution assumes a number of
high-dimensional samples, while the noise vector required for the
initial diffusion model latents is a single vector. The solution is
to divide the latent vector into a number of non-overlapping blocks,
each of which is treated as a sample from the hCLWE distribution. The
blocks {\em must each be large enough to ensure undetectability, but small
enough to ensure enough samples to recover the signal} even if some
blocks are lost or corrupted by noise. As we will see in
Section~\ref{sec:practical-parameters}, for a latent vector of
dimensions $4 \times 64 \times 64$ we choose a block size of $2
\times 4 \times 4$ for our evaluations, but for the moment we leave
this as an implementation parameter.

\paragraph{Frequency Domain.} Using hCLWE samples directly does not give robust recovery of the secret embedded. In order to improve robustness to
perturbations and noise during recovery, we use a standard technique in image watermarking
(see~\cite{al2007combined}): we apply the hCLWE samples in the
frequency domain of the latents, i.e., we take the inverse discrete
wavelet transform (IDWT) of the hCLWE samples to arrive at the latent
vector used to generate the image. The process for image generation
is shown in Figure~\ref{fig:steganography}. The signal recovery
process reverses this process, but requires additional steps in order
to extract the signal in the presence of noise and perturbations.

\paragraph{Recovery.} We now have an approach that allows us to embed
an hCLWE distribution in the latent space of a diffusion model, which
we then use to produce an image. The question becomes: how do we
determine if an image contains a signal? From recent
work~\cite{hong2024exact,song2020denoising}, we can invert the
diffusion model to get an estimate of the original latent vector
(assuming the image is not perturbed)---even without the original
prompt or embedding vector used to create the image, albeit with some
noise. Dividing the latent vector into blocks is then trivial, but
the resulting blocks are subject to noise and perturbation. Given
these noisy blocks and the secret direction of the hCLWE
distribution, how do we determine if the secret signal is present?

\begin{figure*}[ht]
\centering
\includegraphics[width=0.9\textwidth,trim={1cm 2cm 1cm 22cm},clip]{figures/blocking.pdf}
\caption{Signal recovery process showing how the latent vector is extracted from the image, divided into blocks, and analyzed using the secret direction.}%
\label{fig:blocking}
\end{figure*}

First, we can measure the error from the ideal hCLWE distribution by
computing the inner product with the secret direction to get a set
of $z_i$ values:

\[
    z_i = \gamma \langle \vec{y}_i, \vec{w} \rangle \bmod 1
\]

In the no-signal case, we expect that for well-chosen parameters
the $z_i$ values will be uniformly distributed in $[0, 1)$. In the
signal case, we expect that the $z_i$ values will be concentrated
near zero, with the strength of the concentration depending on
$\beta$, noise, and perturbations. This produces the situation shown
in Figure~\ref{fig:clwe_error_rose}, which shows the result of a
statistical simulation comparing $z_i$ values calculated from samples
drawn from a standard normal, hCLWE, and hCLWE with noise from a
normal distribution of width $0.2$.

\begin{figure}
\centering
\includegraphics[width=0.9 \columnwidth, trim={1cm 1cm 1cm 1cm}]{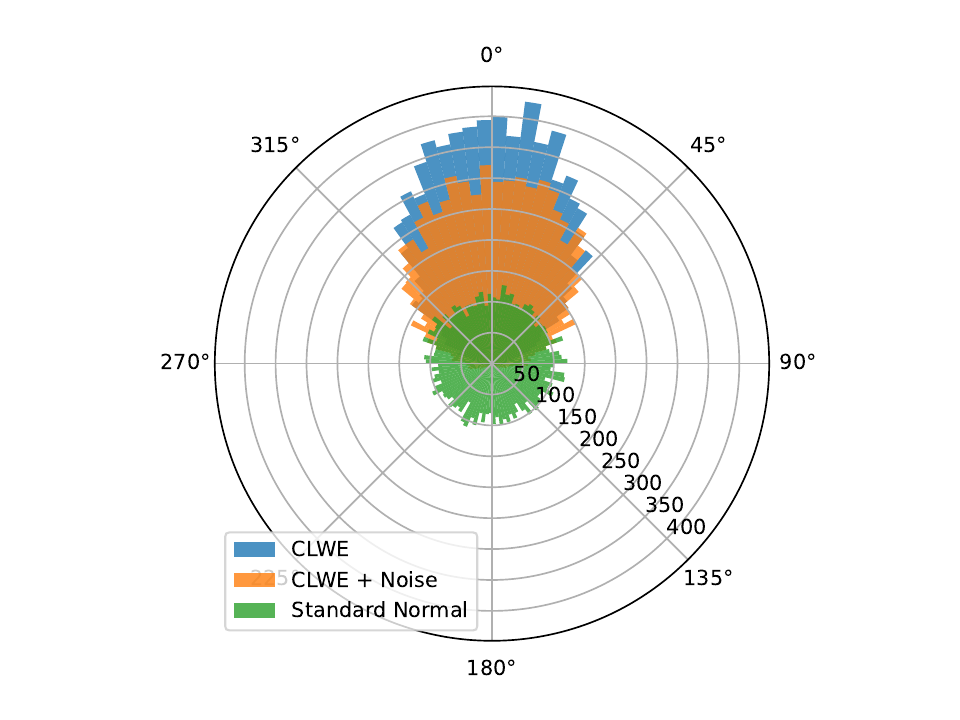}
\caption{Rose diagram of $z_i$ values from statistical simulation of
$10,000$ samples showing the distribution of errors for different
cases, where a $z$ score of $1.0$ is mapped to $360\degree$.}%
\label{fig:clwe_error_rose}
\end{figure}

We now need to decide between two hypotheseses: the null hypothesis
is that the samples are from a standard normal distribution, and thus
the $z_i$ values are uniformly distributed in $[0, 1)$, while the
alternative hypothesis is that the samples are from the hCLWE
distribution, and thus the $z_i$ values have some a concentration
around some (unimodal) peak. Fortunately, the Rayleigh
test~\cite{rayleigh_test} is a well-known test in circular statistics
that provides the optimal ability to distinguish between these two
cases, and provides a $p$-value that indicates the probability that
the null hypothesis is true given the observed data. Thus, if the
$p$-value is below a certain threshold, we can reject the null
hypothesis and conclude that the samples are from the hCLWE
distribution, and thus the image contains a signal. For example, in
the statistical simulation shown in Figure~\ref{fig:clwe_error_rose},
if we restrict to the first 256 samples (as we would expect in our
use case), the Rayleigh test gives a $p$-value of $36.5\%$ for the
standard normal samples, $1.2 \times 10^{-94}$ for the hCLWE samples,
and $5.2 \times 10^{-69}$ for the hCLWE samples with noise.

The resulting recovery process is shown in Figure~\ref{fig:blocking}.
In terms of parameters, we set $n$ (the number of dimensions per CLWE
sample), $\gamma$ and $\beta$ (CLWE parameters) for security of the
reductions in Section~\ref{sec:practical-parameters}, and $m$ is
the size/dimensions of the latent vector divided by $n$.

\section{Security Proof}%
\label{sec:security_proof}

In this section, we provide the formal definitions and proofs of completeness,
soundness, and undetectability of \tool{}.

\subsection{Formal Definitions}%
\label{sec:formal_defs}

Define:

\begin{equation}
    \rho_s(\vec{x}) = \exp\left(-\pi \| \vec{x} / s \|^2 \right)
\end{equation}

Where $\vec{x} \in \reals^n$. Note that $\rho_s(\vec{x})/s^n$ is the
PDF of the gaussian distribution with covariance $s^2/2\pi \cdot
I_n$.

\begin{definition}[hCLWE Distribution]%
\label{def:hclwe_dist}

    Let $\vec{w} \in \reals^n$ be a unit vector, and $\beta,
    \gamma > 0$ be parameters. Define the homogenous continuous
    learning with errors (hCLWE) distribution
    $H_{\vec{w},\beta,\gamma}$ over $\reals^n$ to have density at
    point $\vec{y}$ proportional to:

    \[
        \rho(\vec{y}) \cdot
        \sum_{k \in \mathbb{Z}} \rho_\beta (k - \gamma \langle \vec{w}, \vec{y} \rangle)
    \]

\end{definition}

\begin{definition}[Decision hCLWE Problem]%
\label{def:hclwe}

Let $\vec{w} \in \reals^n$ be a unit vector drawn uniformly at
random, and $\beta, \gamma > 0$ be parameters. Let $\mathcal{N}$ be
an oracle that responds to queries with samples from $\mathcal{N}(0,
I_n)$, and let $\mathcal{H}_{\vec{w},\beta,\gamma}$ be an oracle that
responds with samples from $H_{\vec{w},\beta,\gamma}$. The decision
hCLWE problem is to distinguish $\mathcal{N}$ and
$\mathcal{H}_{\vec{w},\beta,\gamma}$. That is for distinguisher $D$,
we say that $D$ solves the decision hCLWE problem with advantage
$Adv_{D, \mathcal{N}, \mathcal{H}_{\vec{w},\beta,\gamma}}$.

\end{definition}

\begin{assumption}[hCLWE]%
\label{def:hclwe_assumption}

No PPT distinguisher $D$ can solve the decision hCLWE problem with
non-negligible advantage when given a set of $m$ samples, where $m$
is at most polynomial in $n$.

\end{assumption}

The above assumption is stricter than the standard CLWE hardness, as it considers an adversary with a limited number of samples.  

\subsection{Undetectability}%
\label{sec:undetectability_proof}

Theorem~\ref{thm:stego_undetectability} gives our main theoretical
result: detecting the watermark is as hard as solving the decision
hCLWE problem. Before we dive into the technical details of the
proof, we first provide an intuitive explanation. Our proof is by
reduction, that is if there exists a distinguisher that can reliably
distinguish watermarked and unwatermarked images, then we can use
that distinguisher to solve the decision hCLWE problem. Recall that
in the hCLWE problem the distinguisher is given an oracle that for
all queries returns samples from either the hCLWE distribution or the
standard normal distribution. In the reduction, whenever the
watermark distinguisher requests an image, we construct the latent
vector from samples from the given oracle, which we give to the
latent diffusion model to generate an image, before passing the image back to
the watermark distinguisher. If the oracle gives us standard normal
samples, then this corresponds to the unwatermarked case, while if
the oracle gives us hCLWE samples, then this corresponds to the
watermarked case. Other than the samples used to build the initial
latent vector, the process is identical and thus the difference in
probability between the two cases can only come from the samples.
Note that if the latent vector has dimension $mn$, then each image
generated requires $m$ samples from the oracle. If the watermark
distinguisher has a non-negligible advantage, then so does our hCLWE
distinguisher.

\begin{theorem}%
\label{thm:stego_undetectability}

Assuming the hCLWE assumption holds for $n$ dimensions with
parameters $\gamma, \beta$, the \tool{} scheme described
in Section~\ref{sec:approach} with latent vectors of dimension $mn$
such that $m$ is at most polynomial in $n$ is undetectable by
Definition~\ref{def:wm_undetectable}.

\end{theorem}

\begin{proof}

    Consider a PPT distinguisher $D^\mathcal{O}$ that has
    non-negligible advantage $\epsilon$ in distinguishing between
    watermarked and non-watermarked images, that is $Adv_{D,
    \mathcal{M}_{\theta}, \mathcal{W}_{ \theta, k}} \ge \epsilon$.

    We now construct $D'$ that solves the decision hCLWE problem,
    given access to $\mathcal{O}'$, which is either $\mathcal{N}$, an
    oracle that returns samples from the standard normal, or
    $\mathcal{H}_{\vec{w},\beta,\gamma}$, an oracle that gives
    samples from the hCLWE distribution\footnote{The oracle is
    consistent for all queries, that is if the oracle is
    $\mathcal{N}$ then all results are standard normal. The oracle
    does change distributions between queries.}. Distringuisher $D'$
    works as follows:

    \begin{enumerate}
        
        \item On receiving a request to $\mathcal{O}$ with query $\pi$, $D'$:
        
        \begin{enumerate}

            \item requests $m$ samples from the oracle, $\vec{y}_1,
            \ldots, \vec{y}_m \leftarrow \mathcal{O}', \vec{y}_i \in
            \reals^n$, where $mn$ is the size of the latent space of
            the model.
        
            \item Assembles the samples to form the dimensions of the
            latent space, i.e.\ $\vec{z}' \leftarrow (\vec{y}_1,
            \ldots, \vec{y}_m) \in \reals^{mn}$.
            
            \item Applies the inverse DWT to the assembled samples to
            obtain a latent vector, i.e.\ $\vec{z} \leftarrow
            \text{IDWT}(\vec{z}')$.
            
            \item Uses the result as the latent vector to sample an
            image from the model, $\vec{x} \leftarrow
            \text{Sample}_{\theta}(\pi, \vec{z})$.

            \item Forwards the image $\vec{x}$ to $D$.

        \end{enumerate}

        \item $D'$ outputs the same result as $D$.

    \end{enumerate}

    If $D$ has advantage $\epsilon$ in distinguishing
    watermarked and unwatermarked images, then $D'$ has at least
    $\epsilon$ advantage in distinguishing hCLWE samples from
    standard normal samples. Since all of the operations in $D'$ are
    polynomial time, this contradicts the hCLWE assumption.
\end{proof}

\subsection{Sampling from hCLWE}%
\label{sec:sampling_hclwe}

Sampling from the hCLWE distribution is non-trivial. In~\cite{clwe}
it is suggested that the hCLWE distribution be sampled by
rejection-sampling from the standard normal distribution. However,
this is not practical for our setup, as not only does this require
drawing many more samples than needed, it also assumes the base
distribution is the standard normal and easy to sample. In our case,
the base distribution is the model's latent space after a DWT
transform. Rather than rejection sampling in this space, we propose a
much simpler approach: beginning with samples of the appropriate base
distribution (e.g.\ standard normal and then apply DWT), project the
samples to the nearest CLWE point and then add noise in the secret
direction. We can then reverse any transforms needed to get the
samples in the required domain.

\begin{algorithm}[t]
\caption{Method to convert samples from a base distribution to hCLWE.}%
\label{alg:hclwe_sample}
\KwIn{Secret direction $\vec{w}$,
    Samples $\vec{y}_1, \ldots, \vec{y}_m \in \reals^{n}$}
\KwOut{Samples $\vec{y}'_1, \ldots, \vec{y}'_m \in \reals^{n}$}
$\gamma' \gets \sqrt{\beta^2 + \gamma^2}$\;
\For{$i \in [m]$}{
    $k_i \gets \lceil \gamma' \langle \vec{y}_i, \vec{w} \rangle \rfloor$\;
    $z_i \gets \mathcal{N}(0, \beta)$\;
    $\vec{y}_i' \gets \vec{y}_i + \left( (z_i + k_i \cdot \gamma / \gamma') / \gamma' - \langle \vec{y}_i, \vec{w} \rangle \right) \vec{w}$\;
}
\KwRet{$\vec{y}'_1, \ldots, \vec{y}'_m$}
\end{algorithm}

Algorithm~\ref{alg:hclwe_sample} shows the details of this method, which
is slightly more complicated due to the fact that the gaussian
mixtures have width $\beta / \gamma' = \beta / \sqrt{\beta^2 +
\gamma^2}$ and are separated by $\gamma / {\gamma'}^2$. If $\beta \ll
\gamma$ then $\gamma' \approx \gamma$ and this simplifies to the
expected steps. Indeed, for practical purposes the difference is
insignificant. This can be seen in Figure~\ref{fig:hclwe_pdf}, as the
blue area is a histogram generated by a simulation of $10,000$
samples using this approximation and is not simply the area under
the dashed line (the expected PDF).

\begin{claim}

The method shown in Algorithm~\ref{alg:hclwe_sample}, on input from
the standard normal distribution $\mathcal{N}$, produces samples from
$H_{\vec{w},\beta,\gamma}$.

\end{claim}

\begin{proof}

    Consider a single input sample $y$ drawn from the standard normal
    $\mathcal{N}$. Let $\gamma' = \sqrt{\beta^2 + \gamma^2}$,
    $\tilde{y} = \langle \vec{y}, \vec{w} \rangle$ and $\vec{y}^\perp
    = \vec{y} - \tilde{y} \cdot \vec{w}$, i.e.\ $\tilde{y}$ is the
    projection of $\vec{y}$ in the secret direction and
    $\vec{y}^\perp$ is the component of $\vec{y}$ perpendicular to
    it. Set $k = \lceil \gamma' \tilde{y} \rfloor$. We can now
    rewrite the output of the algorithm as:
    \[
        \vec{y}' = \vec{y}^\perp +
        \left(\frac{z + k \cdot \gamma / \gamma'}{\gamma'}\right) \vec{w}
        = \vec{y}^\perp + \tilde{y}' \vec{w}
    \]
    
    From this we can write the PDF of $\vec{y}'$ as the convolution of
    the choices of $k$ and $z$:
    \[
        \rho(\vec{y}^\perp) \cdot \sum_{k \in \mathbb{Z}}
        \rho_{\gamma'}(k) \rho_{\beta}(z)
    \]
    By substituting out $z$, we get:
    \[
        \rho(\vec{y}^\perp) \cdot \sum_{k \in \mathbb{Z}}
        \rho_{\gamma'}(k)
        \rho_{\beta / \gamma'}\left(\tilde{y}' - k \gamma / {\gamma'}^2 \right)
    \]

    This expression can be manipulated algebraically to achieve the
    same expression as in Definition~\ref{def:hclwe_dist} (see
    Appendix~\ref{sec:hclwe_algebra}).
\end{proof}
\section{Evaluation}%
\label{sec:eval}

Returning to the problem statement in Section~\ref{sec:problem} and
specifically the model shown in Figure~\ref{fig:wm_use_cases}, we
evaluate \tool{} against the following questions:

\begin{enumerate}[start=1,label={\bfseries RQ\arabic*:}, align=left]
    \item What practical CLWE parameters should be used for \tool{}?

    \item How well can we recover the \tool{} watermark after inverting the raw images?

    \item What is the quality of the images generated by \tool{}?
    
    \item Do steganographic attacks work against \tool{}?

    \item How robust is \tool{} to benign transformations, such as JPEG compression, on the watermarked images?   

\end{enumerate}

\sourceref{}

\subsection{RQ1. Practical Parameters and an Adaptive Covariance Attack}%
\label{sec:practical-parameters}

The CLWE problem is relatively recent candidate for a cryptographic
hardness assumption, and as such, optimal parameters for security are
not fully settled. The original reductions from hard problems to
CLWE~\cite{clwe} require that $\gamma \ge 2 \sqrt{n}$ and a
non-trivial value of $\beta > 0$, using a quantum reduction.
In~\cite{clwe_hardness}, it is further confirmed using a classical
reduction that $\gamma = \tilde{\Omega}(\sqrt{n})$ is required for
theoretical security. Using the stronger assumption of subexponential
hardness of LWE, it is also shown in~\cite{clwe_hardness} that
$\gamma = {\left(\log n\right)}^{\frac{1}{2} + \delta} \log \log n$
should be sufficient, for arbitarily small $\delta > 0$. All of these
are shown to be sufficient \emph{asymptotically}. The question
becomes: what concrete values of $n, \gamma$, and $\beta$ should we
use in practice? In particular, image sampling and inversion
processes add significant noise, so it is diserable to use the
smallest possible $\gamma$ and $\beta$ to ensure that we can recover
the watermark, while maintaining undetectability.

To answer this question, we extended the covariance attack described
in~\cite{clwe} to produce an adaptive attack that can work with fewer
samples. In this attack, the adversary computes the covariance matrix
of a large set of samples, computes the eigenvalues, and finds the
maximum difference with $1/2\pi$, see
Algorithm~\ref{alg:covariance_attack} for an explicit description. In
the original attack, the adversary checks if this score is greater
than $\gamma^2 \exp(-\pi(\beta^2 + \gamma^2))$. While this attack is
successful when large number of samples is available, i.e.\ $m >
O\left(\exp(\gamma^2)\right)$, the threshold given is not useful for
smaller sample sizes. For practical usage, we are interested in
parameters that are secure for adversaries that have access to a more
limited number of samples, i.e.\ $m$ at most polynomial in $\gamma$
or $n$. To that end, we adjust the procedure slightly to compare the
maximum eigenvalues in the covariance matrix produced by a set of
CLWE samples and normally distributed samples and use the Area Under
the Curve (AUC) of the Receiver Operating Characteristic (ROC) as a
measure of the attack's success. This produces an attack that
can be successful with fewer samples, at the expense of possibly more false positives.

\begin{algorithm}[t]
\caption{Covariance method to determine if samples are from a CLWE distribution or standard normal.}%
\label{alg:covariance_attack}
\KwIn{Samples $A \in \reals^{m \times n}$}
\KwOut{Covariance score}
$\Sigma_m \gets \frac{1}{2 \pi m} A^T A$\;
$\mu_1, \ldots, \mu_n \gets \text{eigenvalues}(\Sigma_m)$\;
\KwRet{$\max_{i \in [n]} \left| \mu_i - \frac{1}{2 \pi} \right|$}\;
\end{algorithm}

For each combination of $n$ (dimensionality), $m$ (number of
samples), and $\gamma$ (gaussian pancake spacing), and fixed $\beta =
0.001$, we produced $100$ iterations using both CLWE samples and
normally distributed samples, and then score both using
Algorithm~\ref{alg:covariance_attack}. We then compute the AUC score,
such that $1.0$ indicates that it is possible to distinguish the two
sets, and $0.5$ indicates that the two sets are indistinguishable. A
score in between indicates the presence of false positives or false
negatives depending on the threshold used. We plot the results in
Figure~\ref{fig:cov_scores}. As expected, with $\gamma = 8$
(theoretically sound for $n \le 64$) the attack is unsuccessful even
with a million samples. At the opposite extreme, with $\gamma = 1$
the attack is successful with several thousand samples.
Interestingly, the attack requires more samples with larger $n$. This
makes sense, as covariance estimation requires linearly more samples
as the number of dimensions increases. If $\gamma$ is increased to
$2$ or $4$, then in each case the number of samples needed for a
successful attack increases an order of magnitude each time.
Also plotted in Figure~\ref{fig:cov_scores} are the
attack accuracies given by the threshold recommended by prior theoretical work~\cite{clwe}. As expected,
for smaller $\gamma$ values the theoretical and practical attacks
converge quickly. For larger $\gamma$ values, the thresholds become
less useful, and were unable to produce any distinguishing power
for any $\gamma$ greater than $1$.

In the CLWE hardness literature~\cite{clwe, clwe_hardness}, the
security parameter is usually defined as the number of dimensions
($n$), with security proven for $\gamma \ge \Omega(\sqrt{n})$,
implying that $\gamma$ should grow with $n$ to maintain security when
the number of samples $m$ is not fixed. However, when the sample size
is fixed, our results show that increasing $n$ while keeping $\gamma$
fixed may actually improve security. In particular, if the total
amount of data ($n \cdot m$) is fixed, then increasing $n$ while
keeping $\gamma$ fixed reduces $m$, which makes the problem harder.

\pgfplotsset{cycle list/Paired-12}
\begin{figure}[!h]
\begin{center}
    \pgfplotstableread[col sep=tab,]{data/cov_attack.tsv}\cov
\begin{tikzpicture}
    \begin{axis}[
        width=0.95*\columnwidth,
        height=6cm,
        legend style={at={(0.5, -0.2)},anchor=north},
        xlabel={Number of Samples ($m$)},
        ylabel={AUC Score},
        xmode=log,
        ymin=0.5, ymax=1.0,
        no markers,
        cycle list/Set1-9,
        mark options={solid},
        thick,
        grid=both,
        grid style={dashed, gray!50},
        legend columns=2, 
        legend style={
            /tikz/column 2/.style={ 
                column sep=5pt, 
            },
        },
        ]

        \addplot table [x={m}, y={n32_g1}]{\cov};
        \addlegendentry{$n=32, \gamma=1$}
        \addplot table [x={m}, y={n64_g1}]{\cov};
        \addlegendentry{$n=64, \gamma=1$}
        \addplot table [x={m}, y={n32_g2}]{\cov};
        \addlegendentry{$n=32, \gamma=2$}
        \addplot table [x={m}, y={n32_g4}]{\cov};
        \addlegendentry{$n=32, \gamma=4$}
        \addplot table [x={m}, y={n32_g8}]{\cov};
        \addlegendentry{$n=32, \gamma=8$}

        \addplot [ dashed ] table [x={m}, y={n32_g1_tacc}]{\cov};
        \addlegendentry{$n=32, \gamma=1$ t. acc.}
        \addplot [ dotted ] table [x={m}, y={n64_g1_tacc}]{\cov};
        \addlegendentry{$n=64, \gamma=1$ t. acc.}
    \end{axis}
\end{tikzpicture}
\caption{Covariance attack scores from statistical simulations. Each
combination of $n, m$ and $\gamma$ was simulated $100$ times for both
CLWE and Normal distributions to produce scores using
Algorithm~\ref{alg:covariance_attack}. The resulting ROC AUC is shown
on the vertical axis. A score of 0.5 is equivalent to random
guessing, and 1.0 indicates that it's possible to perfectly
distinguish the two categories. The theoretical attack accuracy ("t.
acc." entries) correspond to the accuracy of using the threshold
value $\gamma^2 \exp(-\pi(\beta^2 + \gamma^2))$ as suggested in the
original work~\cite{clwe}.}%
\label{fig:cov_scores}
\end{center}
\end{figure}

In the context of diffusion model images, the latent space is
typically $64 \times 64 \times 4 = 2^{14}$. Thus, if we divide the
latent space into blocks of size at least $n \geq 32$, this produces
$m \leq 512$ samples per image, well under 1,000 samples. Based on
the results of our test above, we conclude that $\gamma$ as small as
$2$ should be sufficient.

\begin{mdframed}[backgroundcolor=mygray]
    With $n \ge 32$, setting $\gamma \geq 2$ and $\beta \geq 0.001$
    should provide practical indistinguishability of CLWE from
    the normal distribution, provided the number of samples is
    limited to $m \leq 1,000$.
\end{mdframed}

We stress that the above parameters are based on our empirical
evaluations and not on the theoretical reductions of CLWE hardness
given in prior works. Further research on tighter reductions and the
concrete parameters to use in practice is important future work. For
this reason we provide theoretical proofs of security that apply to
any choice of $\gamma$ and $\beta$, but use the parameters from this
section in order to show a proof-of-concept.

\subsection{Evaluation Setup}%
\label{sec:evalsetup}

We perform our evaluation on state-of-the-art text to image open
sourced diffusion models, specifically, Stable Diffusion version 2.1.
We generate $512 \times 512$ images starting with latent embeddings
of size $4 \times 64 \times 64$ in $50$ steps. We use the inversion
process of~\cite{hong2024exact} for obtaining latent vector estimates
from iamges.

\paragraph{Datasets.} We use two standard benchmark datasets: Stable
Diffusion Prompts (SDP)\footnote{\url{https://huggingface.co/datasets/Gustavosta/Stable-Diffusion-Prompts}}, and
COCO~\cite{lin2014microsoft}. We generate an image without a watermark and for
each watermarking technique using the same prompt and initial seed,
for each of the first $100$ prompts in each dataset.

\paragraph{Baselines.} \treering{} and \gs{} are the two state-of-the-art
pre-processing based watermarking techniques that claim some level of
undetectability unlike the post-processing based methods which are
known to be detectable~\cite{treering}. We use these two methods as our
baselines. We use the ``ring'' pattern for \treering{}, and the
standard \gs{} method, but note that we fixed the key and nonce for all
images. The original code for \gs{} generated a fresh key and nonce for
each image, which is not practical since the recovery process would
not have access to the key and nonce for a specific image. All parameters
for \treering{} and \gs{} match the defaults published in the respective papers.

\paragraph{Evaluation Metrics.} We use Frechet Inception Distance
(FID)~\cite{heusel2017gans} to measure the quality of generated images. FID
measures the difference in the embeddings between two images using
the Inception model, which is a heuristic for how similar they are.
Measuring the FID between the watermarked images and their
unwatermarked counterparts gives a sense of the quality degradation
of the watermarking scheme, i.e.\ a lower FID implies better quality
images.

All of the watermarking schemes produces a score indicative of the
likelihood of the watermark being present in the image. Thus, the
performance of the watermarking scheme depends on choosing a
threshold for this score. This produces a true positive rate (TPR,
corresponds to completeness in Definition~\ref{def:wm_complete}) and
false positive rate (FPR, corresponds to soundness in
Definition~\ref{def:wm_sound}) for a given threshold. By plotting the
TPR vs FPR curve for varying thresholds we obtain the Receiver
Operating Characteristic (ROC) curve, a commonly used method for
evaluating the performance of binary classifiers. The area under the
ROC curve (AUC) is a metric that summarizes the performance of the
classifier. A score of 1 indicates perfect performance, while a score
of 0.5 indicates performance no better than random. Thus we use the
AUC score to evaluate the performance of the watermarking schemes.

\paragraph{System Specification.} We ran our experiments on an AMD
EPYC 7443P 4GHz processor with 96 GB of RAM and 4 Nvidia A40 GPUs
with 45 GB of memory each, running CUDA 12.2 and Nvidia drivers
535.129.03.

\subsection{RQ2. Watermark Recovery}%
\label{sec:eval-soundness-completeness}

We answer how sound and complete \tool{}'s recovery strategy is by
producing a watermarked and unmarked image for each of the $100$
prompts. As discussed in Section~\ref{sec:practical-parameters} we
set $\gamma = 2, \beta = 0.001$ with a block size of $2 \times 4
\times 4$ an apply the watermark in the DWT domain. We then apply
\tool{}'s recovery strategy to all the images and compare the
results. We find that for both datasets, \tool{} is able to
distinguish between watermarked and unmarked images with high
accuracy, with an AUC score greater than 0.99 in both cases. For
comparison, the baseline techniques produced AUC scores of $1.0$ for
the same conditions.

\begin{mdframed}[backgroundcolor=mygray]

    \tool{}'s detector has high soundness and completeness
    (AUC scores $\geq 0.99$) in unperturbed images.

\end{mdframed}

\subsection{RQ3. Image Quality}%
\label{sec:quality}

As can be seen in Figure~\ref{fig:image_fid}, \tool{} generates
images with the closest quality to the non-watermarked images for
both datasets, while \gs{} and \treering{} produces an FID distance
significantly greater (worse) than \tool{}.

The degradation in quality of the images indicates that both the
baselines are not undetectable, especially when many watermarked
images are analyzed at once for detectability. This is expected for
\treering{} since it does not embed provably undetectable watermarks.
It is also expected for \gs{} as it only proves undetectability for a
single image. We explain this in detail in our discussion regarding
\gs{} in Section~\ref{sec:eval-steganographic}.

\begin{mdframed}[backgroundcolor=mygray]

    \tool{} generates images with similar quality to non-watermarked
    images and significantly higher quality images than the
    baselines.

\end{mdframed}

\begin{figure}
\centering
\begin{tikzpicture}
\begin{axis}[
    ybar,
    height=5.25cm,
    x=0.35 \columnwidth,
    bar width=0.09 \columnwidth,
    enlarge x limits={abs=0.2 \columnwidth},
    legend style={at={(0.5,-0.15)},
      anchor=north,legend columns=-1},
    ylabel={FID Score},
    symbolic x coords={COCO, SDP},
    xtick=data,
    nodes near coords,
    nodes near coords align={vertical},
    ymin=0,
    ]
\addplot coordinates {(COCO,144.0364051012072) (SDP,142.79526935232798)};
\addplot coordinates {(COCO,180.99061408122213) (SDP,172.67958443347743)};
\addplot coordinates {(COCO,40.299619303914255) (SDP,42.13922105086317)};
\legend{Tree Ring,Gaussian Shading,CLUE-Mark}
\end{axis}
\end{tikzpicture}
\caption{Image quality as measured by FID to unwatermarked images on two standard datasets.
Lower numbers indicate higher quality (closer to unwatermarked).}%
\label{fig:image_fid}
\end{figure}

\subsection{RQ4. Resilience to Steganographic Attacks}%
\label{sec:eval-steganographic}

A steganographic attack attempts to remove the watermark from marked
images by averaging the difference between the watermarked and
unwatermarked images over many samples. In order to compare with
prior work and confirm the theory behind our construction, we
reproduced the results from~\cite{steg_attack} using $100$ pairs of
watermarked and unwatermarked images generated for each watermarking
method, and average the difference between the images. As can be seen
in the examples in Figure~\ref{fig:steg_images}, there are obvious
patterns in the averages from \treering{} and \gs{}, but \tool{} produces a
nearly blank image. We then subtract this average from the
watermarked images and attempt to recover the watermark. As shown in
Figure~\ref{fig:steg_auc}, the attack successfully removes the \treering{}
and \gs{} watermarks, but is unable to remove the \tool{} watermark. We
note that the resulting watermark-removed images for \treering{} are of
similar quality, but for \gs{} suffer some degradation. We expect that
a more sophisticated attack could remove the \gs{} watermark without
degrading the image.

\begin{mdframed}[backgroundcolor=mygray]

    \tool{} is robust to standard steganographic attacks unlike the
    baselines whose watermark can be detected and removed.

\end{mdframed}

\begin{figure}
\centering
\begin{tikzpicture}
\begin{axis}[
    ybar,
    height=5.25cm,
    x=0.35 \columnwidth,
    bar width=0.09 \columnwidth,
    enlarge x limits={abs=0.2 \columnwidth},
    legend style={at={(0.5,-0.15)},
      anchor=north,legend columns=-1},
    ylabel={AUC},
    symbolic x coords={COCO, SDP},
    xtick=data,
    nodes near coords,
    nodes near coords align={vertical},
    ]
\addplot coordinates {(COCO,0.508) (SDP,0.427)};
\addplot coordinates {(COCO,0.510) (SDP,0.581)};
\addplot coordinates {(COCO,0.977) (SDP,0.980)};
\legend{Tree Ring,Gaussian Shading,CLUE-Mark}
\end{axis}
\end{tikzpicture}
\caption{AUC scores after steganographic attack.
Values closer to 1 indicate presence of the watermark after the attack,
while closer to 0.5 indicate removal of the watermark.}%
\label{fig:steg_auc}
\end{figure}

\begin{figure*}
\centering
\fbox{\includegraphics[width=0.78\textwidth,trim={0 0 15cm 0},clip]{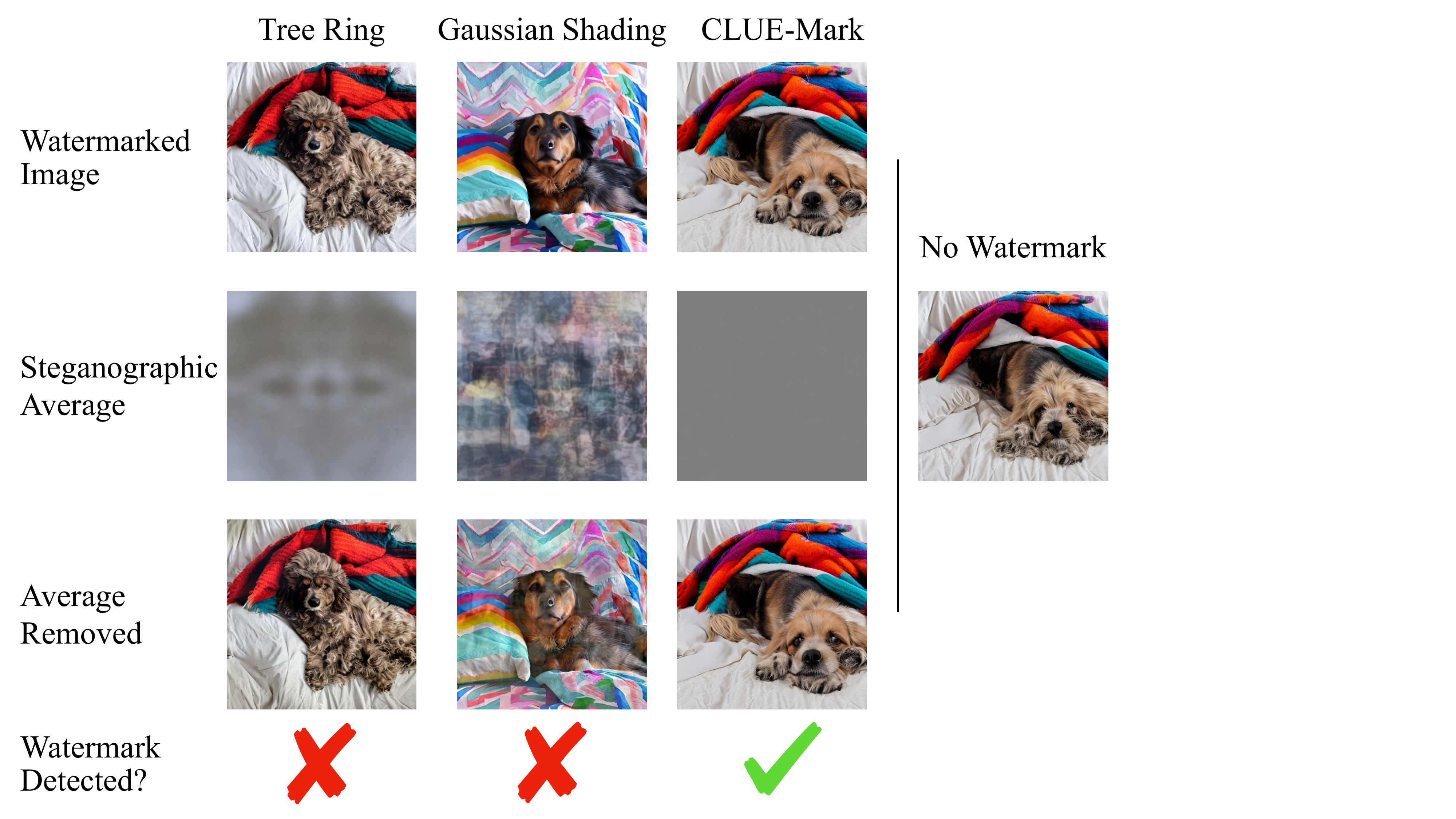}}
\caption{Steganographic attack examples against related schemes. All images generated using
the prompt ``a dog laying next to some colorful blankets on a white
bed'' with the same random seed. Gaussian shading is evaluated with a fixed nonce.}%
\label{fig:steg_images}
\end{figure*}

\subsection{RQ5. Robustness to Benign Transformations}%
\label{sec:eval-robustness}

We evaluate the robustness of \tool{} and the baseline techniques
under different non-adversarial post-processing techniques.
Specifically, we apply JPEG compression (quality value $0.95$),
brightness (level 1), and random cropping (0.9 in each direction),
and rotations to the watermarked images and measure recovery
performance. Our results are shown in Figure~\ref{fig:robust}. For
comparison, the baselines all produced AUC scores at or near 1.0 for
all conditions. We observe that \tool{} is somewhat robust to JPEG
compression and brightness changes, but not to the other
perturbations. Thus \tool{} is robust to minor perturbation, but any
significant modifications make the watermark unrecoverable.

\begin{mdframed}[backgroundcolor=mygray]

    \tool{} is robust to minor perturbations such as JPEG compression
    and brightness but not robust to larger perturbations.

\end{mdframed}

\begin{figure}
\centering
\pgfplotstableread[col sep=tab,]{figures/robust.tsv}\robust
\begin{tikzpicture}
\begin{axis}[
    width=0.9*\columnwidth,
    height=6cm,
    ybar,
    bar width=0.08 \columnwidth,
    enlargelimits=0.3,
    legend style={at={(0.5,-0.15)},
      anchor=north,legend columns=2,},
    ylabel={AUC Score},
    xtick=data,
    xticklabels from table={\robust}{Label},
    nodes near coords,
    nodes near coords align={vertical},
    ]

    \addplot table[x expr=\coordindex,y=COCO]{\robust};
    \addlegendentry{COCO}
    \addplot table[x expr=\coordindex,y=SDP]{\robust};
    \addlegendentry{SDP}
\end{axis}
\end{tikzpicture}
\caption{Robustness to common perturbations as measured by RoC AUC for \tool{}.
Values closer to 1.0 indicate watermark presence after perturbations, while
closer to 0.5 indicates removal.}%
\label{fig:robust}
\end{figure}

\section{Related Work}%
\label{sec:related}

Watermarking images has been extensively studied since the
1990s~\cite{van1994digital}. Early research primarily focused on
watermarking individual images, while more recent advancements target
AI-generated images. Unlike traditional images, AI-generated content
allows watermarking at various stages of the generation pipeline,
including altering training data or fine-tuning models to embed
watermarks.

\paragraph{Individual Images.} Post-processing
techniques are commonly applied to individual images, embedding
watermarks by modifying images after creation. Classical methods
operate in either the spatial~\cite{van1994digital, tsai2000joint} or
frequency domain~\cite{al2007combined, guo2002digital,
hamidi2018hybrid, kundur1997robust, lee2007reversible}. Spatial
domain techniques modify pixel intensities, often by altering the
least significant bits of selected pixels. While efficient, these
techniques lack robustness against simple transformations like
brightness adjustments or JPEG compression. In contrast, frequency
domain methods embed watermarks by altering an image's frequency
coefficients after applying a transformation such as the Discrete
Cosine Transform (DCT) or Discrete Wavelet Transform (DWT). These
approaches offer greater robustness to perturbations, while ensuring
that the watermark remains imperceptible, often measured by peak
signal-to-noise ratio (PSNR).

More recently, many deep learning-based techniques have been
introduced in order to enhance robustness and
quality~\cite{jia2021mbrs, tancik2020stegastamp, zhu2018hidden,
liu2019novel, ahmadi2020redmark}. These techniques typically leverage
encoder-decoder neural architectures, where the encoder embeds a
message into an image, and the decoder recovers the message. The
models are trained to minimize differences between the watermarked
and original images for even an adversarial classifier, while
ensuring message recovery even in the presence of noise. Various
architectures, noise injection strategies, and training objectives
have been explored to achieve robustness against diverse
transformations.

However, existing post-processing methods lack provable guarantees of
undetectability and may degrade image quality~\cite{treering,
fernandez2023stable}. Additionally, using the same watermark across
multiple images can expose vulnerabilities to steganographic attacks,
allowing the watermark to be extracted~\cite{steg_attack}.

\paragraph{Generative Models.} Generative models like Generative
Adversarial Networks (GANs) and Latent Diffusion Models have become
standard for producing high-fidelity, ultra-realistic images. When
traditional watermarking methods are applied to AI-generated images,
they often reduce visual quality~\cite{fernandez2023stable}, which is
undesirable.

Recent research addresses this issue by embedding watermarks directly
within the generative process~\cite{yu2020responsible,
yu2021artificial, fernandez2023stable, treering, gaussian_shading,
lukas2023ptw}. One approach involves training parts of the model to
embed messages during generation, allowing extraction via a
pre-trained decoder~\cite{yu2020responsible, fernandez2023stable,
yu2021artificial, lukas2023ptw}. However, these techniques require
multiple training rounds, which is computationally expensive, especially
as the size of state-of-the-art models continues to grow. While
these methods claim improved image quality over post-processing
techniques, they rely on ad-hoc metrics like PSNR and are
significantly more resource-intensive.

\paragraph{Watermarking Diffusion Models.} Recent advances in watermarking latent diffusion model images propose
modifying only the input latent vectors, offering computational
efficiency without altering model parameters~\cite{treering,
gaussian_shading}. Diffusion models then generate images normally,
unaware of the modification to the latents. The watermarks are
detected by reversing the diffusion process and examining latent
vector. Techniques such as Tree Ring~\cite{treering} and Gaussian
Shading~\cite{gaussian_shading} claim state-of-the-art performance in
both visual quality, in the sense that they are imperceptible to
human observers, and robustness, in the sense that after common image
perturbations the watermark can still be recovered. However, none of
these techniques provide provable undetectability; \treering provided
no proof of undetectability only that it was claimed to be
``imperceptible'', while the proof of \gs undetectability requires a
per-image nonce that must be available at verification time, an
assumption that is difficult to realize in practice. Moreover, both watermarks
can still be identified, removed, or forged by steganographic attacks~\cite{steg_attack}.

In work concurrent to ours, \cite{prc_watermarks} proposes using
Pseudo-Random Codes (PRC) to bias the noise input to a diffusion
model, similar to \gs but without requiring a nonce. As the name
suggests, PRCs are a new cryptographic primitive that can be
considered a combination of encryption and error-correcting codes.
Cyphertexts are generated by encoding a message with a secret key,
which can then be decoded even if some subset of the bits are
flipped. Without the secret key the cyphertext appears random to all
computationally-bounded observers. Intuitively, the signs of the
latent noise are chosen to match a PRC encoding of a 0-length
message. By the pseudorandom nature of the PRC, this is equivalent to
choosing the signs randomly, and therefore does not affect the model
output. Similar to \treering and \gs, the model is inverted to obtain
an estimate of the original noise vector, from which the sign values
are extracted and checked using the PRC decoder. If the decoder
outputs a valid message, the image is considered watermarked. In a
sense, one could see \tool{} as a continuous-domain
PRC using CLWE, in that with the key the presence of the code can
still be detected despite noise, while without the
key the samples are indistinguishable from a Gaussian.
Interestingly, in the appendix of~\cite{prc_watermarks}, the authors
describe that in order to achieve robustness to perturbations they
had to use parameters that are not theoretically secure. The parameters picked in \tool{} also do not match those of in prior theoretical reductions, but for a completely different reason: We are building on the security of CLWE against a sample-constrained adversary, which has not been theoretically studied thus far to our knowledge.

Overall, our work advances these approaches by introducing a watermarking scheme that is based on CLWE, which is an alternative cryptographic construct. We aim for computational undetectability is our main goal, while addressing robustness against non-adversarial transformation (such as image compression) that is common when transmitting images over the web or social networks.
Our tool,
\tool{}, achieves good image quality in latent diffusion model
outputs yet watermarks cannot be identified or forged without the
secret key in our empirical evaluation. 
Ours is a novel
application of CLWE in the context of generative models, as prior research
using CLWE focused on implanting
undetectable backdoors in random Fourier feature-based classification models~\cite{goldwasser2022planting}.

\section{Conclusions and Future Work}%
\label{sec:conclusions}

In this work we introduced \tool{}, a provably undetectable
watermarking scheme for diffusion model generated images, leveraging
the cryptographic hardness of the Continuous Learning With Errors
(CLWE) problem. We first defined formally what it means for a
watermark scheme to be undetectable, and then described how \tool{}
works. We then proved that \tool{} is undetectable under the hCLWE
assumption with limited samples. Finally, we evaluated empirically the hCLWE parameter
choices against common attacks, and the performance of \tool{} on
common diffusion models and datasets, showing that it is usable in
practice, produces images with much higher fidelity to the original
unwatermarked images than existing techniques, and cannot be detected
by steganographic techniques.\@ \tool{} is robust to some common
perturbations, such as JPEG compression and brightness changes. The parameters we chose provide a practical security
level.

While \tool{} presents the first demonstration of how CLWE-based watermarks are already usable for diffusion models, it can be improved. Future work can consider improving the robustness of \tool{} to
perturbations. The concrete parameters for CLWE are not yet
well studied, and it is possible that better parameters could improve
the performance of \tool{} while still maintaining undetectability. Finally,
the technique used for inverting the model to obtain the estimate of
the original latent vector is the current bottleneck in recovering
the watermark in terms of both quality and speed. Improvements to the
inversion method could also significantly improve the performance of
\tool{}.



\bibliographystyle{ACM-Reference-Format}
\bibliography{paper}

\appendix
\section{Are Perfect Watermarks possible?}%
\label{sec:robust}

Ideally, a watermark should be robust in the sense that perturbations that do not affect the image quality should not remove the watermark.
Unfortunately, it is impossible for a watermark to simultaneously satisfy completeness, soundness, provable undetectability, and
perfect robustness. We formalize the statement and provide a proof.

\paragraph{Undetectability.} Recall that in
Definition~\ref{def:wm_undetectable} we defined undetectability as
the property that no PPT algorithm can distinguish between a
watermarked image and its unwatermarked counterpart, without access
to the watermarking key $k$.
In other words, the {\em quality} of the 
distribution of the images before and after watermarking does not change, from the perspective of any {\em quality oracle that does not have access to keys and is computationally bounded}.
Lemma~\ref{lemma:imposs-1} bounds the distance between distributions observed under two keys $k,k'$ against any such quality oracle. 
Informally, this tells us that the quality
of images generated from watermarking, even with different keys, is
the same to any downstream function without the watermarking keys.

\begin{lemma}
\label{lemma:imposs-1}

For an undetectabile watermarking scheme, no PPT distinguisher $D$
can distinguish between $\mathcal{W}_{ \theta, k}$ and $\mathcal{W}_{
\theta, k'}$ with non-negligible advantage. That is, $Adv_{D,
\mathcal{W}_{\theta, k}, \mathcal{W}_{ \theta, k}} \le
negl(\lambda)$, assuming $D$ has no access to the keys $k$ or $k'$.

\end{lemma}

\begin{proof}

The proof is by a standard hybrid argument. Consider any PPT
distinguisher $D$. From the definition of undetectability, we have
that $Adv_{D, \mathcal{M}_{\theta}, \mathcal{W}_{ \theta, k}} \le negl(\lambda)$ and 
$Adv_{D, \mathcal{M}_{\theta} , \mathcal{W}_{ \theta, k'}} \le
negl(\lambda)$. By the triangle inequality:

\begin{align*}
    Adv_{D, \mathcal{W}_{\theta, k}, \mathcal{W}_{ \theta, k'}}
    &\le Adv_{D, \mathcal{M}_{\theta}, \mathcal{W}_{ \theta, k}}
    + Adv_{D, \mathcal{M}_{\theta}, \mathcal{W}_{ \theta, k'}} \\
    &\le 2 \cdot negl(\lambda) = negl(\lambda).
\end{align*}

\end{proof}

\paragraph{Perfect Robustness.} If we allow arbitrary perturbations
of the image, it is trivial to remove the watermark by simply
replacing the image with another completely different image. What may
not be immediately apparent is that even if we restrict to
perturbations that maintain image quality, it is still impossible to
have a watermark that is robust to such perturbations and is
undetectible. Here we define quality in the same way as undetectability:
the quality of two sets of images is the same if no PPT algorithm can
distinguish the distributions with non-negligible advantage.

A watermarking scheme is perfectly robust if a watermarked image,
after perturbations that do not change its quality, is treated the
same as the unperturbed image by the (keyed) verifier, i.e., the
$Extract$ function accepts the perturbed image if the watermarked
image before perturbation was accepted. More formally, a {\em
perfectly robust watermark} is one wherein given an $x \leftarrow
\text{Mark}(k,\pi)$, no PPT adversary can find an image $x'$ of the
same quality as $x$ but for which $\text{Extract}(k, x) \neq
\text{Extract}(k, x')$.

\begin{theorem}
\label{thm:impossible}
No watermarking scheme satisfying soundness, completeness, undetectability,
and perfect robustness exists.
\end{theorem}
\begin{proof}
    
Assume for the sake of contradiction that a watermarking scheme satisfying
all four properties exists. We can instantiate the scheme
with a key $k\leftarrow Setup(1^\lambda)$ and by computing $x
\leftarrow Mark(k,\pi)$ for an input $\pi$. By the {\em completeness}
property, the $Extract(k, x)$ must accept. Now consider the
following generic attack strategy: The adversary who does not know
$k$ generates a fresh key $k'\leftarrow Setup(1^\lambda)$ and
generates $x' \leftarrow Mark(k',\pi)$. Note that by
Lemma~\ref{lemma:imposs-1}, PPT algorithms without access to $k$ and
$k'$ cannot distinguish the distributions of $x$ and $x'$ with
non-negligible probability. Thus, $x$ and $x'$ have the {same
quality}. We analyze the behavior of the original verifier, i.e., the
distribution of the random variable $Extract(k, x')$: \begin{itemize}
\item If the verifier $Extract(k,x')$ accepts with probability more
than $negl(\lambda)$, then it violates {\em soundness}. This is
because $k' \neq k$ with high probability $1 - negl(\lambda)$, and
with that probability, we have a PPT adversary who does not know $k$
but has found $x'$ that is accepted by the verifier with the key $k$.

\item If the verifier $Extract(k,x')$ rejects with probability more than $negl(\lambda)$, then the scheme violates {\em perfect robustness}---the adversary has efficiently
\footnote{The adversary breaks the hardness criterion while sticking to the constraint stated for perfect robustness, i.e., knowing $k'$ but not $k$.}
found $x'$ of same quality
\footnote{Note that quality is determined by oracles with no access to keys.}
as $x$ (by Lemma~\ref{lemma:imposs-1}) without knowing $k$, and yet we have $Extract(k,x') \neq Extract(k,x)$. 
\end{itemize}
We have derived a contradiction to the assumptions in both cases above, which are exhaustive and mutually exclusive. Thus, no watermarking scheme satisfies all 4  properties.
\end{proof}

The attack strategy underlying our proof of Theorem~\ref{thm:impossible} is realizable within practical threat models, such as those assumed in recent attacks on LLM watermarking schemes~\cite{freelunchllmwatermarking}.
\section{HCLWE Algebra}%
\label{sec:hclwe_algebra}

First we notice the following equivalence:
\begin{align*}
    &\left(\frac{k}{\gamma'}\right)^2
        + \left(\frac{\tilde{y}' - k \gamma / {\gamma'}^2}{\beta / \gamma'}\right)^2
    \\
    &= \frac{k^2}{\beta^2 + \gamma^2} + \frac{\beta^2 + \gamma^2}{\beta^2}
      \left(\left(\tilde{y}'\right)^2 - \frac{2 \gamma \tilde{y}' k}{\beta^2 + \gamma^2}
            + k^2 \frac{\gamma^2}{\left(\beta^2 + \gamma^2\right)^2}\right)
    \\
    &= \left(\tilde{y}'\right)^2 + \frac{\left(\gamma \tilde{y}'\right)^2}{\beta^2}
        - \frac{2 \left(\gamma \tilde{y}'\right) k}{\beta^2}
        + \frac{k^2}{\beta^2}
    \\
    &= \left(\tilde{y}'\right)^2 + \left(\frac{\gamma \tilde{y}' - k}{\beta}\right)^2
\end{align*}

Where in the second step the last term comes from noticing that:
\begin{align*}
    \frac{k^2}{\beta^2 + \gamma^2} + \frac{\gamma^2}{\beta^2} \frac{k^2}{\beta^2 + \gamma^2}
    = \frac{k^2}{\beta^2 + \gamma^2} \left(1 + \frac{\gamma^2}{\beta^2} \right)
    = \frac{k^2}{\beta^2}
\end{align*}

We can now rearrange the statement in the claim as:
\begin{align*}
    &\rho(\vec{y}^\perp) \cdot \sum_{k \in \mathbb{Z}}
        \rho_{\gamma'}(k)
        \rho_{\beta / \gamma'}\left(\tilde{y}' - k \gamma / {\gamma'}^2 \right)
    \\
        &= \exp\left(-\pi \| \vec{y}^\perp \|^2 \right)
            \cdot \sum_{k \in \mathbb{Z}}
            \exp\left(-\pi \left(\frac{k}{\gamma'}\right)^2 \right)
            \exp\left(-\pi \left(\frac{\tilde{y}' - k \frac{\gamma}{{\gamma'}^2}}{\beta / \gamma'}\right)^2 \right)
    \\
        &= \sum_{k \in \mathbb{Z}} \exp\left[-\pi\left(
            \| \vec{y}^\perp \|^2
            + \left(\frac{k}{\gamma'}\right)^2
            + \left(\frac{\tilde{y}' - k \frac{\gamma}{{\gamma'}^2}}{\beta / \gamma'}\right)^2
        \right)\right]
    \\
        &= \sum_{k \in \mathbb{Z}} \exp\left[-\pi\left(
            \| \vec{y}^\perp \|^2
            + \left(\tilde{y}'\right)^2 + \left(\frac{\gamma \tilde{y}' - k}{\beta}\right)^2
        \right)\right]
    \\
        &= \rho(\vec{y}') \sum_{k \in \mathbb{Z}}
            \rho_\beta (k - \gamma \langle \vec{w}, \vec{y}' \rangle)
\end{align*}

Where in the last step we used the orthogonality of $\vec{y}^\perp$
and $\vec{w}$ to write $\|\vec{y}'\|^2 = \|\vec{y}^\perp\|^2 +
\left(\tilde{y}'\right)^2$, and the definition of $\tilde{y}'$.

\end{document}